\documentclass[letterpaper, 10 pt, conference]{ieeeconf}  

\IEEEoverridecommandlockouts                              
\usepackage{graphics} 
\usepackage{epsfig}
 
\usepackage{amsthm}
\usepackage{xcolor}
\usepackage{times} 
\usepackage{amsmath} 
\usepackage{amssymb}  
\usepackage{float}
\usepackage{optidef}
\usepackage{cite}

\newtheorem{theorem}{Theorem}[section]
\newtheorem{lemma}[theorem]{Lemma}
\newtheorem{proposition}[theorem]{Proposition}

\newtheorem{assumption}{Assumption}
\newtheorem{example}{Example}
\newcommand{\NE}{\mathrm{NE}}
\newcommand{\OC}{\mathrm{OC}}
\title{\LARGE \bf
Robust Information Design with Heterogeneous Beliefs in Bayesian Congestion Games
}

\author{Yuwei Hu, Bryce L. Ferguson
\thanks{Yuwei Hu is a Ph.D. student with the Department of Electrical and Computer Engineering,
Thayer School of Engineering at Dartmouth College, Hanover, NH, USA
        {\tt\small Yuwei.Hu-2.TH@dartmouth.edu}}%
\thanks{Bryce L. Ferguson is Faculty in the Department of Electrical and Computer Engineering,
Thayer School of Engineering at Dartmouth College, Hanover, NH, USA
        {\tt\small Bryce.L.Ferguson@dartmouth.edu}}%
}

\begin{document}

\maketitle
\thispagestyle{empty}
\pagestyle{empty}

\begin{abstract}
In many engineered systems, agents make decisions under incomplete information, creating opportunities for a planner to influence decentralized behavior through signaling. We study how such signaling can be designed in parallel-network, affine latency congestion games when users may not interpret recommendations using the same beliefs assumed by the planner. To do so, we consider Bayesian congestion games with private recommendations and formulate a robust information design problem in which obedience must hold uniformly over a neighborhood of a nominal prior. This addresses the previously uncharacterized issue of whether obedience itself remains reliable under belief heterogeneity, rather than only under the single prior used at the design stage. We characterize policy-level robustness radii, identify regimes in which the robust obedience region remains nonempty, and analyze the resulting robustness--performance tradeoff through a robust value function whose optimal cost is monotone in the robustness requirement and whose local sensitivity is governed by the active obedience constraints.
\end{abstract}

\section{Introduction}

Large-scale engineered systems often involve many agents making decisions under incomplete information. Transportation networks provide a natural example: drivers do not a priori observe the realized traffic state, yet their ultimate routing decisions are influenced by their expected travel times.
Through these routing decisions, drivers interact by congesting roads and collectively affect overall system performance~\cite{beckmann1956transportation,roughgarden2002selfish}. More broadly, similar informational asymmetries arise across many socio-technical systems in which large populations make partially informed decisions~\cite{sandholm2010population}.
An emerging mechanism for coordinating collective behavior involves a well informed operator sharing recommendations, platform guidance, or other partial signals to agents who do not have direct access to the underlying state of the world. Such informational asymmetries create an opportunity for a planner or platform to shape behavior through strategic signaling, often called information design~\cite{kamenica2011bayesian,bergemann2016information,cianfanelli2023information}.

A central challenge in information design, and social influencing control more broadly, is that users need not interpret messages using the same beliefs assumed by the planner. In practice, an information signaling policy may be designed relative to a nominal prior, while at deployment users hold heterogeneous beliefs about the state~\cite{alonso2016bayesian}. In traffic routing settings, this issue is especially important as users must decide whether to comply with a recommendation without directly observing realized network conditions, and compliance itself can depend on how the recommendation is interpreted~\cite{chorus2009traveler}. A policy that is interpreted correctly under the nominal prior used at the design stage may therefore fail to remain compliant when evaluated under nearby but different beliefs. This creates a robustness question that is specific to information design: how much belief heterogeneity can be tolerated before obedience is lost, and what is the cost of designing policies that remain reliable under such heterogeneity?

To study these questions, we consider information design in Bayesian congestion games with private recommendations~\cite{cianfanelli2023information,ferguson2024signaling}. A planner observes the realized state of nature and recommends routes to users, who then decide whether to comply. In the case that all users comply, the policy is termed \emph{obedient}. We focus on the case in which the signaling policy is designed relative to a nominal prior, but users may assess compliance under nearby beliefs. This motivates a notion of \emph{robust obedience}, under which incentive compatibility must hold uniformly over a neighborhood of the nominal prior. The resulting robust information design problem asks not only whether a signaling policy performs well, but whether it remains implementable when users interpret the same recommendations through slightly different beliefs. Our model builds on the classical congestion and routing literature~\cite{Wardrop1952} and on game-theoretic models of congestion effects and equilibrium behavior~\cite{rosenthal1973class}, as well as on recent work on information design in nonatomic routing games~\cite{Zhu2022}.

Information design and persuasion have been studied extensively in economics, control, and game theory~\cite{kamenica2011bayesian,bergemann2016information,chremos2024mechanism}. Recent work has examined signaling in Bayesian congestion and routing
games~\cite{cianfanelli2023information,ferguson2024signaling,Zhu2022,
Wu2017,gould_information_2023,Massicot2022,MassicotLangbort2025}, while other papers study heterogeneous priors in persuasion models~\cite{alonso2016bayesian,HuWeng2021} and robustness or belief perturbation in routing environments~\cite{verbree2024inferring,como2013robust}. Recent work also considers learning-based robust persuasion under uncertainty about receiver belief formation~\cite{BangMalikopoulos2025}. Together, these works highlight that both information structure and heterogeneous beliefs can significantly shape equilibrium outcomes in networked environments. However, these works do not characterize when obedience itself remains robust under heterogeneous beliefs in congestion networks, nor do they connect this question to both the feasibility of obedient signaling and the planner's objective value.

In this work, we study robust obedience in Bayesian congestion games under heterogeneous beliefs along two complementary dimensions. First, for a fixed signaling policy, we characterize a robustness radius within which obedience is preserved, and at the design level we study when the robust obedience region remains nonempty. This includes a benchmark regime with robustness for all uncertainty radii, as well as constrained settings with explicit certified ranges. Second, we study the planner's robustness--performance tradeoff through a robust value function, showing that the value of the robust information design problem is monotone in the robustness requirement and that its local sensitivity is governed by the active obedience constraints. These results provide a tractable framework for quantifying how belief heterogeneity affects both the implementability and the cost of information design in Bayesian congestion games.

\section{PROBLEM SETUP}

\subsection{Bayesian Congestion Games}
\label{subsec:problem_setup_A}

To model traffic routing with possibly random and uncertain delays, we consider a class of Bayesian congestion games in which edge latency functions are random variables~\cite{cianfanelli2023information}.
Let $G = (V,E)$ denote a directed network; in this work, we consider parallel networks, or those where a finite number of directed edges connect a source node to a destination node, i.e., $V = \{s,d\}$ and $E=\{1,\dots,|E|\}$.
Without loss of generality, we assume unit total demand and a nonatomic population~\cite{kamenica2011bayesian}. Aggregate flows can therefore be identified with elements of the simplex \(\Delta(E)\). Let \(f\in\Delta(E)\) denote a flow, where \(f_e\) is the mass routed on edge \(e\in E\).

As the demand on an edge increases, each user traveling via that edge experiences increased cost due to congestion.
In traffic, for example, greater numbers of vehicles sharing a road segment reduce driving speed and thus increases travel latency.
However, the exact congestion characteristics of a road segment vary based on exogenous conditions such as weather or prior traffic accidents. Accordingly, we let \(\Omega\) denote a finite set of states of nature, and assume that a state \(\omega \in \Omega\) is drawn from a prior distribution
\(
\mu_0 \in \Delta(\Omega),
\)
which determines the network's congestion characteristics.

To model the realized congestion costs, we consider that each edge $e \in E$ is endowed with a latency function

\begin{equation}
\label{eq:affine_latency}
\ell_e^\omega(f_e) = a_e^\omega f_e + b_e^\omega,
\qquad
a_e^\omega > 0,\; b_e^\omega \ge 0.
\end{equation}
which, for each $\omega \in \Omega$, is non-negative and non-decreasing to capture the increased cost at an edge caused by congestion.
The assumption that $\ell_e$ is affine is a frequent starting point in the congestion games literature, where affine latency is often used as a baseline model. Here, \(b_e^\omega\) can be interpreted as the free-flow travel time, while \(a_e^\omega\) captures the sensitivity of delay to congestion. It is also reasonable in some traffic settings when focusing on a local operating regime, where delay increases approximately linearly with flow. This is consistent with empirical evidence in moderate-congestion regimes~\cite{kreindler2024peak}.

When users route themselves so as to minimize their own expected cost, the resulting emergent behavior is a \emph{Nash equilibrium} (or, synonymously, a user equilibrium) of the game induced by the expected latency function on each edge~\cite{Wardrop1952,como2013stability,ferguson2024signaling}. In particular, a flow \(f^{\NE}\in\Delta(E)\) is called a Nash equilibrium flow if no user assigned to a positively used edge can strictly decrease cost by unilaterally deviating to another edge. Equivalently, for every edge \(e\in E\),
\(
f^{\NE}_e>0
\Longrightarrow
\ell_e(f^{\NE}_e)\le \ell_{e'}(f^{\NE}_{e'}),
 \forall e'\in E.
\)

To improve this performance, research has often studied how monetary incentives can be used to alter users' decisions and promote Nash equilibria of lower system cost.
In this work, we consider information signaling mechanisms as a means to alter collective behavior, whereby the asymmetry between the system operator's and the system users' knowledge of the realized state of nature $\omega$ introduces the opportunity to selectively reveal information so as to shape posterior beliefs and resulting flows in the network.

\subsection{Information Design}
\label{subsec:information_design}

To understand the role of information asymmetry in our setting, suppose the system operator directly observes the realized state \(\omega \in \Omega\), while the users do not. Instead, the users possess a belief over the realization of \(\omega\). In the nominal design problem, the planner takes this belief to be \(\mu_0 \in \Delta(\Omega)\), from which \(\omega\) is drawn.
After observing the realization of $\omega$, the operator sends private recommendations to
the users to allow them to form new, posterior beliefs.
Upon receiving a recommendation, each user acts by selecting their path to travel along.
We distinguish between a recommendation and an action by
writing
\[
r \in E \quad \text{(recommendation)},
\qquad
a \in E \quad \text{(action)}.
\]
Thus, for an individual user, $r$ is the edge recommended by the operator, while $a$ is the edge ultimately chosen by the user.
We consider the private signaling setting, where the system operator can provide different, private recommendations to different users.

Our focus on recommendation-based signaling policies is without loss of
generality~\cite{bergemann2016information}. Since the designer is allowed to commit to an information
structure and we do not impose any additional exogenous restriction on
the class of feasible signal structures, one may restrict attention to
direct policies in which the signals sent to users are identified with
action recommendations. In other words, rather than considering an
arbitrary signal space together with a separate decoding step, it is
enough to work directly with recommendation rules. This simplification
does not reduce the set of attainable outcomes for the designer; it only
removes additional signal structure that is immaterial for the analysis.
To preserve anonymity across users, we consider a signaling rule that
does not depend on user identity. Formally, the operator commits to a signaling
policy $(\pi,X)$, where
\begin{equation}
\label{eq:policy_definition}
\pi = \{\pi_\omega\}_{\omega\in\Omega},
\qquad
\pi_\omega \in \Delta(X),
\end{equation}
and $X\subseteq\mathbb{R}_+^{E}$ is a feasible recommendation space.
Here, $X$ is a set of feasible flow profiles. Each
$x\in X$ represents a flow the operator may wish to induce and satisfies
\(
x=(x_e)_{e\in E},\)
and
\(
\sum_{e\in E} x_e = 1.
\)

In the unrestricted case, one could simply take $X=\Delta(E)$, so that
every feasible flow over edges is available as a recommendation profile~\cite{Zhu2022,myerson1988mechanism}.
However, we allow for the possibility that some flows are inadmissible
from the designer's perspective, and therefore work with a general
recommendation set $X\subseteq \Delta(E)$.
In Section~\ref{sec:robustness}, we show that if $X$ is a polytope the information design problem can be solved over a finite support, on which we focus much of our analysis.
In general, $X$ is not assumed to be finite, and each
$\pi_\omega$ is a probability measure over
$X$.

The recommendation mechanism is as follows: given $\omega$, the
operator first samples a recommendation profile $x\in X$ according to
$\pi_\omega$. Conditional on $x$, each user independently receives a
recommended edge $r$ according to the weight (and with probability) $x_r$. Therefore, given $\omega$ is realized, the
probability that a user receives recommendation $r=e$ is
\begin{equation}
\label{eq:recommendation_probability}
\mathbb{P}[r=e\mid \omega]
=
\int_X x_e\, d\pi_\omega(x).
\end{equation}

In designing an effective signaling policy, the system operator seeks
to choose \(\pi\) so that no user is incentivized to deviate from
the recommendation they receive. Policies with this property are termed
\emph{obedient}. Since the signaling policy is common knowledge, a user
who receives recommendation \(r\) can compare the expected cost of
following that recommendation with the expected cost of deviating to any
other edge \(a\in E\). By the revelation principle for information design,
it is without loss of generality to restrict attention to direct
recommendation policies satisfying the corresponding obedience
constraints, since any outcome implementable by an arbitrary signaling
scheme can be implemented by a direct recommendation rule that induces
the same outcome distribution and is obedient on path
\cite{kamenica2011bayesian,bergemann2016information,zu2021learning}.

For a signaling policy $(\pi,X)$ and a belief $\mu\in\Delta(\Omega)$,
define the obedience slack for each ordered pair $(r,a)$ with $r\neq a$
by
\begin{equation}
\label{eq:obedience_slack}
A_{r,a}(\pi;\mu,X)
:=
\sum_{\omega\in\Omega}
\mu(\omega)\int_X
x_r\bigl(\ell_a^\omega(x_a)-\ell_r^\omega(x_r)\bigr)\,d\pi_\omega(x).
\end{equation}
This quantity measures the expected cost difference between deviating to
action $a$ and obeying recommendation $r$, under signaling policy $\pi$,
belief $\mu$, and recommendation space $X$.

Let
\(
\mathcal J:=\{(r,a)\in E\times E:\ r\neq a\}
\)
denote the set of ordered deviation pairs. We say that $\pi$ is
obedient under belief $\mu$ on recommendation space $X$ if
\mbox{\(
A_{r,a}(\pi;\mu,X)\ge 0
\)} for all \((r,a)\in\mathcal J.
\)
Under the full-support nominal prior
\(\mu_0\in\operatorname{ri}\Delta(\Omega)\), the corresponding
obedience region is
\begin{equation}
\label{eq:nominal_region}
\OC(\mu_0;X)
=
\left\{
\pi :
A_{r,a}(\pi;\mu_0,X)\ge 0,\ \forall (r,a)\in\mathcal J
\right\}.
\end{equation}

Given a signaling policy $\pi$, the planner's expected system cost
under the nominal prior $\mu_0$ and recommendation space $X$ is
\begin{equation}
\label{eq:cost_function}
C(\pi;\mu_0,X)
:=
\sum_{\omega\in\Omega}\mu_0(\omega)\int_X
\sum_{e\in E}x_e\,\ell_e^\omega(x_e)\,d\pi_\omega(x),
\end{equation}
where $\sum_{e\in E}x_e\,\ell_e^\omega(x_e)$ is the social cost
induced by recommendation profile $x$ in state $\omega$.

The nominal information design problem is thus
\begin{equation}
\label{eq:nominal_problem}
V^*(\mu_0;X)
:=
\min_{\pi\in\OC(\mu_0;X)} C(\pi;\mu_0,X).
\end{equation}
That is, under the prior $\mu_0$, the planner chooses a signaling policy
that minimizes expected system cost subject to the obedience constraints.
The primary focus of this work is addressing how users' possessing incorrect and heterogeneous beliefs on the distribution of the state of nature affects the information design problem in congestion routing.
\subsection{Heterogeneous Beliefs and Robust Obedience}
\label{subsec:robust_obedience}

We now allow for the possibility that users' beliefs over the distribution of $\omega$ differ from the
prior $\mu_0$. The point here is not
that the planner is trying to estimate a hidden ``true'' prior.
Instead, the issue is that users may interpret the same signaling rule
through beliefs that are different from the one used in the planner's
design model.
To capture the uncertainty around users' beliefs, we consider a neighborhood of feasible beliefs
around $\mu_0$. For a given radius $\varepsilon\ge 0$, define
\begin{equation}
\label{eq:uncertainty_set}
U_\varepsilon(\mu_0)
:=
\left\{
\mu\in\Delta(\Omega):\|\mu-\mu_0\|\le \varepsilon
\right\},
\end{equation}
where \(\|\cdot\|\) is any norm on \(\mathbb{R}^{|\Omega|}\).
Throughout, \(\|\cdot\|_*\) denotes the dual norm associated with
\(\|\cdot\|\). Unless otherwise specified, the results below hold for an
arbitrary choice of norm.

The robust obedience region is the set of signaling policies $\pi$ that remain obedient under any belief $\mu \in U_\varepsilon(\mu_0)$, or
\begin{equation}
\label{eq:robust_region}
\overline{\OC}_\varepsilon(\mu_0;X)
:=
\bigcap_{\mu\in U_\varepsilon(\mu_0)} \OC(\mu;X).\nonumber
\end{equation}
In particular,
\(\overline{\OC}_0(\mu_0;X)=\OC(\mu_0;X)\).

In general, each user in the population may have their own belief.
In Proposition~\ref{prop:heterogeneous_closure_short} we show a direct interpretation of robust obedience
in a heterogeneous population.
\begin{proposition}[Robust obedience under heterogeneous beliefs]
\label{prop:heterogeneous_closure_short}
Suppose a heterogeneous population of agents is indexed by $\lambda \in [0,1]$,
and each agent $\lambda$ has belief $\mu^\lambda \in U_\varepsilon(\mu_0).$
If $\pi\in \overline{\OC}_\varepsilon(\mu_0;X)$,
then every individual agent $\lambda \in [0,1]$ is obedient under $\pi$, and therefore any
aggregate flow induced by such a heterogeneous population is obedient
as well.
\end{proposition}
The proof appears in the appendix.

Imposing robust obedience leads to the \emph{robust information design problem}
\begin{equation}
\label{eq:robust_problem}
V^*(\varepsilon;\mu_0,X)
:=
\min_{\pi\in\overline{\OC}_\varepsilon(\mu_0;X)} C(\pi;\mu_0,X).
\end{equation}
Thus, \(V^*(\varepsilon;\mu_0,X)\) is the minimum expected system cost
achievable by a signaling policy that remains obedient for all beliefs
in \(U_\varepsilon(\mu_0)\).
As we consider $X$ to be a given set of constraints,
for notational simplicity, in the remainder of the paper we suppress
the explicit dependence of these objects on \(X\) whenever no confusion
can arise.

The paper develops two main sets of results. Section~\ref{sec:robustness}
studies the robust feasible region and identifies when obedient
signaling remains possible under belief heterogeneity.
We show that
robust feasibility is preserved for all uncertainty radii when the state-wise Nash equilibria are present in $X$, and we characterize explicit certified robustness ranges in
the constrained case. 
Section~\ref{sec:sensitivity} studies how the
robust value function \(V^*(\varepsilon;\mu_0,X)\) changes as the
robustness requirement is tightened. We establish monotonicity of the
value function in $\varepsilon$ and derive a local sensitivity bound, at
differentiability points, in terms of the active obedience constraints
and their conditioning. Together, these results characterize both the
feasibility and the cost of robustness in information design in Bayesian congestion games.

\section{Feasibility of Robust Obedience}
\label{sec:robustness}

This section studies when the robust obedience region
\(\overline{\OC}_\varepsilon(\mu_0)\) remains nonempty and how much
prior heterogeneity can be tolerated before feasibility is lost. For
brevity, we write
\[
\rho(\pi):=\sup\left\{\varepsilon\ge 0:
\pi\in\overline{\OC}_\varepsilon(\mu_0;X)\right\},
\]
as the robust radius of a policy $\pi$ and
suppress the dependence on \(\mu_0\) and \(X\) whenever no confusion
can arise.

We are therefore interested in the largest uncertainty radius that can
still be supported by some signaling policy, namely
\begin{equation}
\label{eq:rho_star_X}
\rho^\star
:=
\sup_{\pi\in\Pi}\rho(\pi).
\end{equation}
Thus \(\rho^\star\) measures the maximal level of belief
misspecification under which the robust obedience region can remain
feasible. We begin with a finite-support reduction, which lets us carry
out the remainder of the analysis on a finite recommendation set.
\begin{proposition}[Finite-support reduction]
\label{prop:finite_support_reduction}
Let \(\pi=\{\pi_\omega\}_{\omega\in\Omega}\) be a signaling policy over
\(X\). Then there exist a finite set \(\widetilde X\subseteq X\) and a
signaling policy
\(\widetilde\pi=\{\widetilde\pi_\omega\}_{\omega\in\Omega}\) with
\(
\widetilde\pi_\omega\in\Delta(\widetilde X),
\forall \omega\in\Omega,
\)
such that
\(
C(\widetilde\pi;\widetilde X)=C(\pi;X),
\)
and
\(
A_{r,a}(\widetilde\pi;\mu)
=
A_{r,a}(\pi;\mu),
\forall (r,a)\in\mathcal J,\ \forall \mu\in\Delta(\Omega).
\)
\end{proposition}

Proposition~\ref{prop:finite_support_reduction} shows that restricting
attention to finite-support signaling policies incurs no loss of
generality for either the nominal design problem or the robust
obedience analysis. The proof is deferred to the Appendix.

Hence, in the remainder of the paper, we work with a fixed finite
recommendation set
\begin{equation}
\label{eq:X_finite_specialized}
X=\{x^{(1)},\dots,x^{(K)}\}\subset \mathbb{R}_+^E,
\end{equation}
where \(x^{(k)}\) denotes the \(k\)-th element of \(X\). To simplify notation, we omit the dependence on \(X\) when it is clear from context.

\subsection{Nash Recommendations and Robust Feasibility}
\label{subsec:nash_special_case}

We first identify a regime in which robust feasibility is never lost.
Suppose that for every state \(\omega\in\Omega\), \(x^{\NE}(\omega)\) is a Nash equilibrium for the game with latency functions $\{\ell_e^\omega\}_{e \in E}$ and that
\(
x^{\NE}(\omega)\in X,
\) for all \(
 \omega\in\Omega.
\)
\begin{proposition}[Nash recommendations imply robust nonemptiness]
\label{prop:nash_nonempty_all_eps}
If \(X\) contains a statewise Nash equilibrium flow $x^\NE(\omega)$ for every
\(\omega\in\Omega\), then
\(
\overline{\OC}_\varepsilon(\mu_0)\neq\varnothing,
\forall \varepsilon\ge 0.
\)
\end{proposition}

\begin{proof}
Define the Nash recommendation policy as
\[
\pi^{\NE}_\omega(x)
=
\mathbf{1}\{x=x^{\NE}(\omega)\},
\qquad
\forall \omega\in\Omega,\ \forall x\in X.
\]
Under the supposition that $X$ contains $x^{\rm NE}(\omega)$ for each $\omega \in \Omega$, the policy is nominally feasible.
Fix any \(\mu\in U_\varepsilon(\mu_0)\) and any ordered pair
\((r,a)\in\mathcal J\). Since \(x^{\NE}(\omega)\) is a Nash
equilibrium flow in state \(\omega\), we have
\[
x_r^{\NE}(\omega)
\Bigl(
\ell_a^\omega(x_a^{\NE}(\omega))
-
\ell_r^\omega(x_r^{\NE}(\omega))
\Bigr)
\ge 0,
\qquad
\forall \omega\in\Omega.
\]
Therefore,
{\small
\[
A_{r,a}(\pi^{\NE};\mu)
=
\sum_{\omega\in\Omega}
\mu(\omega)\,x_r^{\NE}(\omega)
(
\ell_a^\omega(x_a^{\NE}(\omega))
-
\ell_r^\omega(x_r^{\NE}(\omega))
).
\]
}
is always non-negative.
Hence \(\pi^{\NE}\in\overline{\OC}_\varepsilon(\mu_0)\), which proves
the claim.
\end{proof}

Thus, if all statewise Nash flows are available as recommendations, the
robust obedience region is nonempty for every \(\varepsilon \ge 0\). 
In general, we are interested in the
constrained case, where \(X\) may not contain all such Nash equilibrium flows.
In the following example, we show that if $X$ does not contain all of the state-wise Nash equilibria, it is possible that robust feasibility is lost.
\begin{example}
Robust feasibility may fail immediately in the constrained case.
Consider \(E=\{1,2\}\), \(\Omega=\{\omega_1,\omega_2\}\), and
\(X=\{x^0\}\) with \(x^0=(1,0)\). Then the only signaling policy is
\(\pi^0\), where \(\pi^0_\omega(x^0)=1\) for all \(\omega\). Let
\[
\ell_1^{\omega_1}(z)=0,\quad \ell_2^{\omega_1}(z)=2,\qquad
\ell_1^{\omega_2}(z)=2,\quad \ell_2^{\omega_2}(z)=1.
\]
Then \(x^{\NE}(\omega_1)=(1,0)\) and \(x^{\NE}(\omega_2)=(0,1)\), so
\(X\) does not contain all statewise Nash equilibrium flows. Take
\(\mu_0(\omega_1)=\frac13\) and \(\mu_0(\omega_2)=\frac23\). For the
only relevant deviation \((1,2)\),
\begin{equation*}
A_{1,2}(\pi^0;\mu)
=
\mu(\omega_1)(2-0)+\mu(\omega_2)(1-2)
=
3\mu(\omega_1)-1.
\end{equation*}
Hence obedience holds if and only if \(\mu(\omega_1)\ge \frac13\), and
at the nominal prior \(A_{1,2}(\pi^0;\mu_0)=0\). For any
\(\varepsilon>0\), one can choose \(\mu\in U_\varepsilon(\mu_0)\) with
\(\mu(\omega_1)<1/3\), which gives \(A_{1,2}(\pi^0;\mu)<0\). Since
\(\pi^0\) is the only feasible policy,
\[
\overline{\OC}_\varepsilon(\mu_0)=\varnothing,\qquad \forall \varepsilon>0.
\]
\end{example}
Thus, without the statewise Nash flows, a positive robustness margin is no longer automatic, though constrained recommendation sets may still admit a strictly positive certified robustness radius.

\subsection{Recommendation Support Patterns}
\label{subsec:support_restricted_certificate}

We now study the constrained regime. Since the exact
policy-level robustness radius $\rho(\pi)$ may be hard to compute
directly, we introduce a simpler certified radius
$\widehat\rho(\pi)\le \rho(\pi)$. This gives a sufficient guarantee for a
fixed signaling policy and will later be used to build a global
feasibility certificate.

For each ordered pair \((r,a)\in\mathcal J\), define
\begin{equation}
\label{eq:realized_cost_difference}
\Delta_{r,a}(\omega,x)
:=
\ell_a^\omega(x_a)-\ell_r^\omega(x_r).
\end{equation}
For a policy \(\pi\in\Pi\), define
\begin{equation}
\label{eq:recommendation_masses}
m_r^\omega(\pi):=
\sum_{x\in X}\pi_\omega(x)x_r,
\quad
m_r(\pi;\mu_0):=
\sum_{\omega\in\Omega}\mu_0(\omega)m_r^\omega(\pi),
\end{equation}
and
\(\operatorname{supp}_r(\pi):=
\{(\omega,x)\in\Omega\times X:\ \pi_\omega(x)x_r>0\}\).
For every \((r,a)\in\mathcal J\) with
\(\operatorname{supp}_r(\pi)\neq\varnothing\), define
\begin{equation}
\label{eq:sigma_bounds_pi}
\begin{aligned}
\underline{\sigma}_{r,a}^{\pi}
&:= \inf_{(\omega,x)\in\operatorname{supp}_r(\pi)}
\Delta_{r,a}(\omega,x),\\
\overline{\sigma}_{r,a}^{\pi}
&:= \sup_{(\omega,x)\in\operatorname{supp}_r(\pi)}
\Delta_{r,a}(\omega,x).
\end{aligned}
\end{equation}
and \(M_{r,a}^{\pi}:=
\max\{|\underline{\sigma}_{r,a}^{\pi}|,\,
|\overline{\sigma}_{r,a}^{\pi}|\}\).

If prior uncertainty is measured in \(\ell_p\), let \(q\) be the dual
exponent, \(1/p+1/q=1\). For any \(\pi\in\Pi\), define
\begin{equation}
\label{eq:policy_certificate_radius}
\widehat{\rho}(\pi)
:=
\min_{\substack{(r,a)\in\mathcal J:\\ m_r(\pi;\mu_0)>0}}
\frac{\underline{\sigma}_{r,a}^{\pi}\,m_r(\pi;\mu_0)}
{|\Omega|^{1/q}M_{r,a}^{\pi}},
\end{equation}
omitting pairs with \(m_r(\pi;\mu_0)=0\).

The next result shows that \(\widehat{\rho}(\pi)\) certifies a nontrivial
range of uncertainty radii for which a fixed signaling policy remains
robustly obedient.

\begin{theorem}
\label{thm:certified_epsilon_range_policy}
Fix \(\pi\in\Pi\). Suppose that \(\underline{\sigma}_{r,a}^{\pi}>0\) for all
\((r,a)\in\mathcal J\) with \(m_r(\pi;\mu_0)>0\). Then
\(\pi\in\overline{\OC}_{\varepsilon}(\mu_0)\) for all
\(\varepsilon\in[0,\widehat{\rho}(\pi)]\). Consequently,
\[
\widehat{\rho}(\pi)\le \rho(\pi).
\]
\end{theorem}

Thus \(\widehat{\rho}(\pi)\) gives an explicit robustness guarantee for
\(\pi\). In particular, the guarantee deteriorates as \(M_{r,a}^{\pi}\)
grows, since larger values of \(M_{r,a}^{\pi}\) correspond to greater
spread in the realized deviation costs under \(\pi\). The proof uses the
following lemma.
\begin{lemma}
\label{lem:support_restricted_bounds}
Fix a policy \(\pi\) and a pair \((r,a)\in\mathcal J\). If
\(\operatorname{supp}_r(\pi)\neq\varnothing\), then for each
\(\omega\in\Omega\),
\[
\underline{\sigma}_{r,a}^{\pi}\,m_r^\omega(\pi)
\le
d_{r,a}^\omega(\pi)
\le
\overline{\sigma}_{r,a}^{\pi}\,m_r^\omega(\pi),
\]
where \(d_{r,a}^\omega(\pi):=
\sum_{x\in X}\pi_\omega(x)x_r\Delta_{r,a}(\omega,x)\).
Therefore, if
\begin{equation}
\label{eq:D_and_A_inner_product}
D_{r,a}(\pi)
=
\bigl(d_{r,a}^\omega(\pi)\bigr)_{\omega\in\Omega},
\qquad
A_{r,a}(\pi;\mu)
=
\langle \mu, D_{r,a}(\pi)\rangle,
\end{equation}
then
\[
\underline{\sigma}_{r,a}^{\pi}\,m_r(\pi;\mu_0)
\le
A_{r,a}(\pi;\mu_0)
\le
\overline{\sigma}_{r,a}^{\pi}\,m_r(\pi;\mu_0),
\]
and
\[
\|D_{r,a}(\pi)\|_q
\le
M_{r,a}^{\pi}\,\|m_r^\cdot(\pi)\|_q
\le
|\Omega|^{1/q}M_{r,a}^{\pi},
\]
where \(m_r^\cdot(\pi):=
\bigl(m_r^\omega(\pi)\bigr)_{\omega\in\Omega}\).
\end{lemma}

\begin{proof}
By \eqref{eq:sigma_bounds_pi}, multiplying the pointwise bounds on
\(\Delta_{r,a}(\omega,x)\) by \(\pi_\omega(x)x_r\ge 0\) and summing over
\(x\in X\) gives
\[
\underline{\sigma}_{r,a}^{\pi}\,m_r^\omega(\pi)
\le
d_{r,a}^\omega(\pi)
\le
\overline{\sigma}_{r,a}^{\pi}\,m_r^\omega(\pi),
\qquad
\forall \omega\in\Omega.
\]
Using \eqref{eq:recommendation_masses} and
\eqref{eq:D_and_A_inner_product}, summing against \(\mu_0(\omega)\)
gives the bounds on \(A_{r,a}(\pi;\mu_0)\). Also,
\(|d_{r,a}^\omega(\pi)|\le M_{r,a}^{\pi}m_r^\omega(\pi)\), so
\[
\|D_{r,a}(\pi)\|_q
\le
M_{r,a}^{\pi}\|m_r^\cdot(\pi)\|_q
\le
|\Omega|^{1/q}M_{r,a}^{\pi}.
\]
\end{proof}

\begin{proof}[Proof of Theorem~\ref{thm:certified_epsilon_range_policy}]
Fix \((r,a)\in\mathcal J\) with \(m_r(\pi;\mu_0)>0\) and let
\(\mu\in U_\varepsilon(\mu_0)\). By \eqref{eq:D_and_A_inner_product},
Hölder's inequality, and Lemma~\ref{lem:support_restricted_bounds},
\[
A_{r,a}(\pi;\mu)
\ge
\underline{\sigma}_{r,a}^{\pi}\,m_r(\pi;\mu_0)
-\varepsilon |\Omega|^{1/q}M_{r,a}^{\pi}.
\]
If \(\varepsilon\le \widehat{\rho}(\pi)\), then
\eqref{eq:policy_certificate_radius} gives
\(A_{r,a}(\pi;\mu)\ge 0\). Pairs with \(m_r(\pi;\mu_0)=0\) are inactive,
so \(\pi\in\overline{\OC}_{\varepsilon}(\mu_0)\) for all
\(\varepsilon\in[0,\widehat{\rho}(\pi)]\).
\end{proof}

Thus \(\widehat{\rho}(\pi)\) gives a certified robustness range for a
fixed signaling policy.
We now shift our focus to identifying robustness certificates for the feasibility of the robust information design problem, i.e., bounds on \eqref{eq:rho_star_X}.
Let
\begin{equation}
\label{eq:rho_hat_star_X}
\widehat{\rho}^{\,\star}
:=
\max_{\pi\in\Pi}\widehat{\rho}(\pi).
\end{equation}
By Theorem~\ref{thm:certified_epsilon_range_policy},
\[
\widehat{\rho}^{\,\star}
=
\max_{\pi\in\Pi}\widehat{\rho}(\pi)
\le
\max_{\pi\in\Pi}\rho(\pi)
=
\rho^\star,
\]
and thus \(\widehat{\rho}^{\,\star}\) is a certified lower bound on
\(\rho^\star\).

Directly maximizing $\widehat{\rho}(\pi)$ over $\Pi$ is computationally difficult because
$\underline{\sigma}_{r,a}^{\pi}$, $\overline{\sigma}_{r,a}^{\pi}$, and
$M_{r,a}^{\pi}$ depend on the realized support of $\pi$, making $\widehat{\rho}(\pi)$ possibly discontinuous over $\Pi$.
We therefore index the problem by realized support patterns.
Once a support pattern is fixed, the remaining optimization over \(\pi\)
is linear.

For each route \(r\), let \(X_r:=\{x\in X:\ x_r>0\}\). A support pattern
is a collection \(S=(S_r)_{r\in E}\) with
\(S_r\subseteq \Omega\times X_r\). Let \(\mathfrak S\) denote the finite
family of support patterns induced by policies in \(\Pi\). For
\(\pi\in\Pi\), let \(S(\pi):=\bigl(\operatorname{supp}_r(\pi)\bigr)_{r\in E}\).

For any \(S\in\mathfrak S\) and any \((r,a)\in\mathcal J\) with
\(S_r\neq\varnothing\), define
\[
\begin{aligned}
\underline{\sigma}_{r,a}(S)
&:= \min_{(\omega,x)\in S_r}\Delta_{r,a}(\omega,x),\\
\overline{\sigma}_{r,a}(S)
&:= \max_{(\omega,x)\in S_r}\Delta_{r,a}(\omega,x),\\
M_{r,a}(S)
&:= \max\{|\underline{\sigma}_{r,a}(S)|,\,
|\overline{\sigma}_{r,a}(S)|\}.
\end{aligned}
\]
Also define
\[
\Gamma(S):=
\{(\omega,x)\in\Omega\times X:\ \exists r\in E \text{ such that } (\omega,x)\in S_r\}.
\]

For each \(S\in\mathfrak S\), let \(\widehat{\rho}^{\,\star}(S)\) be the
optimal value of
{\setlength{\jot}{0.2pt}
\begin{maxi*}|s| {\pi,\rho}{\rho}{}{}
\addConstraint{\underline{\sigma}_{r,a}(S)\,m_r(\pi;\mu_0)}
{\ge \rho\,|\Omega|^{1/q}M_{r,a}(S)}{}
\addConstraint{}{}{\quad \forall (r,a)\in\mathcal J:\ S_r\neq\varnothing}
\addConstraint{\sum_{x\in X}\pi_\omega(x)}{=1,}
{\quad \forall \omega\in\Omega}
\addConstraint{\pi_\omega(x)}{\ge 0,}
{\quad \forall \omega\in\Omega,\ \forall x\in X}
\addConstraint{\pi_\omega(x)}{=0,}
{\quad \forall (\omega,x)\in(\Omega\times X)\setminus \Gamma(S).}
\end{maxi*}
}
where the constraints enforce consistency of \(\pi\) with the prescribed
support pattern \(S\).
\vspace{-0.1em}
\begin{theorem}
\label{thm:pattern_decomposition}
\(\widehat{\rho}^{\,\star}
=
\max_{S\in\mathfrak S}\widehat{\rho}^{\,\star}(S)\).
Hence maximizing the certified radius reduces to finitely many linear
programs, one for each support pattern.
\end{theorem}

\begin{proof}
Because \(\Omega\) and \(X\) are finite, so is \(\mathfrak S\).
Fix \(\pi\in\Pi\), and let \(S(\pi)\in\mathfrak S\) be its realized
support pattern. Then, for every \((r,a)\in\mathcal J\) with
\(m_r(\pi;\mu_0)>0\),
\[
\underline{\sigma}_{r,a}^{\pi}
\hspace{-1pt}=
\underline{\sigma}_{r,a}(S(\pi)),
\
\overline{\sigma}_{r,a}^{\pi}
\hspace{-1pt}=
\overline{\sigma}_{r,a}(S(\pi)),
\
M_{r,a}^{\pi}
\hspace{-1pt}=
M_{r,a}(S(\pi)).
\]
Moreover, if \(\pi_\omega(x)>0\), then \((\omega,x)\in\Gamma(S(\pi))\),
so \((\pi,\widehat{\rho}(\pi))\) is feasible for the program defining
\(\widehat{\rho}^{\,\star}(S(\pi))\). Hence
\[
\widehat{\rho}(\pi)
\le
\widehat{\rho}^{\,\star}(S(\pi))
\le
\max_{S\in\mathfrak S}\widehat{\rho}^{\,\star}(S).
\]
Taking the maximum over \(\pi\in\Pi\) yields
\[
\widehat{\rho}^{\,\star}
\le
\max_{S\in\mathfrak S}\widehat{\rho}^{\,\star}(S).
\]

Conversely, fix \(S\in\mathfrak S\), and let \(\pi^S\) be an optimizer
of the corresponding linear program. Let \(S':=S(\pi^S)\). Then
\(S'_r\subseteq S_r\) for every route \(r\). Therefore, for every
\((r,a)\in\mathcal J\) with \(m_r(\pi^S;\mu_0)>0\),
\[
\underline{\sigma}_{r,a}^{\pi^S}
=
\underline{\sigma}_{r,a}(S')
\ge
\underline{\sigma}_{r,a}(S),
\
M_{r,a}^{\pi^S}
=
M_{r,a}(S')
\le
M_{r,a}(S).
\]
Thus
\(
\widehat{\rho}^{\,\star}(S)
\le
\widehat{\rho}(\pi^S)
\le
\widehat{\rho}^{\,\star}.
\)
Taking the maximum over \(S\in\mathfrak S\) gives the 
claimed equality.
\end{proof}

Thus \(\widehat{\rho}^{\,\star}\) is a computable certified lower bound
on \(\rho^\star\). In particular, if
\(0\le \varepsilon<\widehat{\rho}^{\,\star}\), then there exists
\(\pi\in\Pi\) such that \(\pi\in\overline{\OC}_{\varepsilon}(\mu_0)\),
and hence \(\overline{\OC}_{\varepsilon}(\mu_0)\neq\varnothing\).
While \(\mathfrak S\) is finite, its cardinality can grow exponentially
with \(|\Omega||X|\) in the worst case; in the numerical examples
considered here, however, the resulting enumeration is tractable.

\section{Sensitivity}
\label{sec:sensitivity}

In the previous section, we studied the feasibility of robust obedience.
We now examine its effect on the planner's objective. For every
\(\varepsilon\) such that
\(\overline{\OC}_{\varepsilon}(\mu_0)\neq\emptyset\), define
\[
V^*(\varepsilon)
:=
\min_{\pi\in\overline{\OC}_{\varepsilon}(\mu_0)} C(\pi).
\]
We first record its monotonicity.

\subsection{Monotonicity}
\label{subsec:sensitivity_monotonicity}

\begin{proposition}
\label{prop:V_monotone}
The value function \(V^*(\varepsilon)\) is monotone nondecreasing on
its feasibility domain. That is, if
\(0\le \varepsilon_1\le \varepsilon_2\) and
\(\overline{\OC}_{\varepsilon_2}(\mu_0)\neq\emptyset\), then
\[
V^*(\varepsilon_1)\le V^*(\varepsilon_2).
\]
\end{proposition}

\begin{proof}
Since \(U_{\varepsilon_1}(\mu_0)\subseteq
U_{\varepsilon_2}(\mu_0)\), every policy that is obedient for all
\(\mu\in U_{\varepsilon_2}(\mu_0)\) is also obedient for all
\(\mu\in U_{\varepsilon_1}(\mu_0)\). Hence
\[
\overline{\OC}_{\varepsilon_2}(\mu_0)
\subseteq
\overline{\OC}_{\varepsilon_1}(\mu_0).
\]
Minimizing the same objective over nested feasible sets yields
\[
V^*(\varepsilon_1)\le V^*(\varepsilon_2).
\]
\end{proof}

Thus, stronger robustness requirements weakly increase the optimal
cost.

\subsection{Local slope via active constraints}
\label{subsec:sensitivity_slope}

The monotonicity result describes only the direction of change. We next
study the local slope. Specialize to Euclidean belief uncertainty and let
\[
P_\Omega:=I-\frac{1}{|\Omega|}\mathbf 1\mathbf 1^\top
\]
be the projection onto \(\{e:\mathbf 1^\top e=0\}\).

For radii satisfying
\[
\varepsilon\le
\min_{\omega\in\Omega}\mu_0(\omega)
\sqrt{\frac{|\Omega|}{|\Omega|-1}},
\]
the exact robust constraints take the form
\[
h_{r,a}(\pi,\varepsilon)
:=
A_{r,a}(\pi;\mu_0)
-\varepsilon\|P_\Omega D_{r,a}(\pi)\|_2
\ge0.
\]
Thus,
\[
\overline{\OC}_\varepsilon(\mu_0)
=
\{\pi:h_{r,a}(\pi,\varepsilon)\ge0,\ 
\forall(r,a)\in\mathcal J\}.
\]
Since \(A_{r,a}\) and \(D_{r,a}\) are linear in the policy probabilities,
these constraints are second-order-cone representable.

\begin{assumption}[Relative Slater]
\label{ass:relative_slater}
At the radius \(\varepsilon\) under consideration, there exists
\(\hat\pi\), with \(\hat\pi_\omega\) in the relative interior of
\(\Delta(X)\) for every \(\omega\in\Omega\), such that
\[
h_{r,a}(\hat\pi,\varepsilon)>0
\qquad
\forall(r,a)\in\mathcal J.
\]
\end{assumption}

Let \(\pi^*(\varepsilon)\) be an optimal solution and define the active
obedience set
\[
B(\varepsilon)
:=
\{(r,a)\in\mathcal J:
h_{r,a}(\pi^*(\varepsilon),\varepsilon)=0\}.
\]
Let
\[
T_\Pi
:=
\left\{
v:
\sum_{x\in X}v_\omega(x)=0\ \ \forall\omega,\quad
v_\omega(x)=0\ \text{if }\pi_\omega^*(x)=0
\right\}
\]
be the tangent space of the minimal face of the policy simplex, and let
\(P_\Pi\) denote its orthogonal projection. Define the projected
active-gradient matrix
\begin{equation}
\label{eq:active_constraint_jacobian}
G_B^\Pi(\varepsilon)
:=
\Bigl[
P_\Pi\nabla_\pi
h_{r,a}(\pi^*(\varepsilon),\varepsilon)
\Bigr]_{(r,a)\in B(\varepsilon)}.
\end{equation}
Also define
\[
\Delta_{\max}
:=
\max_{(r,a)\in\mathcal J}
\max_{\omega\in\Omega}
\max_{x\in X}
\left|
x_r\bigl(
\ell_a^\omega(x_a)-\ell_r^\omega(x_r)
\bigr)
\right|.
\]

\begin{theorem}
\label{thm:slope_bound_structure}
Fix a radius \(\varepsilon\) satisfying the bound above and suppose
Assumption~\ref{ass:relative_slater} holds at this radius. Suppose
\(V^*(\varepsilon)\) is differentiable, every active constraint is
differentiable at \(\pi^*(\varepsilon)\), and
\(B(\varepsilon)\neq\varnothing\). Let
\(\lambda^*(\varepsilon)\) be an optimal multiplier vector and define
\[
\lambda_B^*(\varepsilon)
:=
\bigl(
\lambda_{r,a}^*(\varepsilon)
\bigr)_{(r,a)\in B(\varepsilon)}.
\]
If
\[
\sigma_{\min}\bigl(G_B^\Pi(\varepsilon)\bigr)>0,
\]
then
\begin{equation}
\label{eq:value_slope_structural_bound}
\begin{aligned}
0\le \frac{dV^*}{d\varepsilon}(\varepsilon)
&=
\sum_{(r,a)\in B(\varepsilon)}
\lambda_{r,a}^*(\varepsilon)
\|P_\Omega D_{r,a}(\pi^*(\varepsilon))\|_2\\
&\le
\sqrt{|\Omega|\,|B(\varepsilon)|}\,
\Delta_{\max}
\frac{
\|P_\Pi\nabla C(\pi^*(\varepsilon))\|_2
}{
\sigma_{\min}(G_B^\Pi(\varepsilon))
}.
\end{aligned}
\end{equation}
\end{theorem}

\begin{proof}
The stated radius bound guarantees that every perturbation \(e\)
satisfying
\[
\mathbf 1^\top e=0,
\qquad
\|e\|_2\le\varepsilon
\]
also satisfies \(\mu_0+e\ge0\). Hence the simplex nonnegativity
constraints do not further restrict the tangent-space ball, yielding
the exact representation \(h_{r,a}(\pi,\varepsilon)\ge0\) used above.

For the Lagrangian
\[
\mathcal L(\pi,\lambda;\varepsilon)
=
C(\pi)
-
\sum_{(r,a)\in\mathcal J}
\lambda_{r,a}h_{r,a}(\pi,\varepsilon),
\]
Assumption~\ref{ass:relative_slater} gives strong duality and the
existence of KKT multipliers. At a differentiability point of the value
function, standard value-function sensitivity gives
\[
\frac{dV^*}{d\varepsilon}(\varepsilon)
=
\frac{\partial\mathcal L}{\partial\varepsilon}
\bigl(
\pi^*(\varepsilon),
\lambda^*(\varepsilon);
\varepsilon
\bigr).
\]
Since
\[
\frac{\partial h_{r,a}}{\partial\varepsilon}
=
-\|P_\Omega D_{r,a}(\pi)\|_2,
\]
we obtain
\[
\frac{dV^*}{d\varepsilon}(\varepsilon)
=
\sum_{(r,a)\in\mathcal J}
\lambda_{r,a}^*(\varepsilon)
\|P_\Omega D_{r,a}(\pi^*(\varepsilon))\|_2.
\]
Complementary slackness eliminates the inactive constraints, giving the
equality in \eqref{eq:value_slope_structural_bound}. In particular, the
slope is nonnegative.

For every \((r,a)\in\mathcal J\),
\[
\|P_\Omega D_{r,a}(\pi^*(\varepsilon))\|_2
\le
\|D_{r,a}(\pi^*(\varepsilon))\|_2
\le
\sqrt{|\Omega|}\,\Delta_{\max},
\]
because each component of \(D_{r,a}\) is a convex combination of the
quantities appearing in the definition of \(\Delta_{\max}\). Therefore,
\[
\frac{dV^*}{d\varepsilon}(\varepsilon)
\le
\sqrt{|\Omega|}\,\Delta_{\max}
\|\lambda_B^*(\varepsilon)\|_1.
\]

The KKT stationarity condition can be written as
\[
0
=
\nabla C(\pi^*(\varepsilon))
-
\sum_{(r,a)\in B(\varepsilon)}
\lambda_{r,a}^*(\varepsilon)
\nabla_\pi h_{r,a}(\pi^*(\varepsilon),\varepsilon)
+n,
\]
where \(n\) belongs to the normal cone of the policy simplex at
\(\pi^*(\varepsilon)\). Since \(T_\Pi\) is the tangent space of the
minimal face containing \(\pi^*(\varepsilon)\), \(P_\Pi n=0\).
Projecting stationarity onto \(T_\Pi\) therefore gives
\[
P_\Pi\nabla C(\pi^*(\varepsilon))
=
G_B^\Pi(\varepsilon)\lambda_B^*(\varepsilon).
\]
Hence, if
\(\sigma_{\min}(G_B^\Pi(\varepsilon))>0\),
\[
\|\lambda_B^*(\varepsilon)\|_2
\le
\frac{
\|P_\Pi\nabla C(\pi^*(\varepsilon))\|_2
}{
\sigma_{\min}(G_B^\Pi(\varepsilon))
}.
\]
Finally,
\[
\|\lambda_B^*(\varepsilon)\|_1
\le
\sqrt{|B(\varepsilon)|}\,
\|\lambda_B^*(\varepsilon)\|_2,
\]
which proves the stated bound.
\end{proof}

This is a local result: nondifferentiability may occur when the active
set changes. The condition
\(\sigma_{\min}(G_B^\Pi(\varepsilon))>0\) is a LICQ-type condition
relative to the minimal face of the policy simplex; it requires the
projected active obedience gradients to be linearly independent. As
\(\sigma_{\min}(G_B^\Pi(\varepsilon))\) becomes small, these gradients
become nearly linearly dependent and the bound permits a large local
sensitivity to the robustness radius. The remaining factors reflect the
magnitude and number of active obedience constraints. If
\(B(\varepsilon)=\varnothing\), then
\[
\frac{dV^*}{d\varepsilon}(\varepsilon)=0.
\]

\section{Simulation}
\label{sec:simulation}

This section illustrates how stronger robustness requirements shrink the
feasible region and can change the optimizer. We first present a minimal
two-state, two-profile geometric example, and then turn to a Monte Carlo
study.

\subsection{Geometric example}
\label{subsec:simulation_geometry}


We illustrate a specific example, where
\[
E=\{1,2,3\},\quad
\Omega=\{\omega_1,\omega_2\},\quad
\mu_0(\omega_1)=\mu_0(\omega_2)=\tfrac12,
\]
and $X=\{x^{(1)},x^{(2)}\}$ with $x^{(1)}=(0.55,\,0.25,\,0.20)^\top$ and $x^{(2)}=(0.25,\,0.55,\,0.20)^\top$.
In this small example, we can write \(p_t:=\pi(x^{(1)}\mid\omega_t)\), \(t=1,2\), and identify the
policy with \(p=(p_1,p_2)\in[0,1]^2\).

Using the affine latency model from \eqref{eq:affine_latency}, let
\[
\begin{array}{c|ccc}
 & e=1 & e=2 & e=3 \\ \hline
\omega_1 & 0.28\,x+0.50 & 0.72\,x+1.28 & 0.46\,x+0.94 \\
\omega_2 & 0.74\,x+1.26 & 0.27\,x+0.52 & 0.44\,x+0.96
\end{array}
\]
so that edge \(1\) is relatively attractive in state \(\omega_1\) and
edge \(2\) is relatively attractive in state \(\omega_2\).



To provide a geometric interpretation to the feasible set of the robust information design problem, we consider the policy space $p \in [0,1]^2$ and illustrate the robust obedience sets.
Figure~\ref{fig:robust_geometry} shows the six nominal boundaries
\(A_{r,a}(p;\mu_0)=0\), the nominal feasible region
\(\OC(\mu_0)\), and the robust obedience regions
\(\overline{\OC}_\varepsilon(\mu_0)\) for increasing
\(\varepsilon\), computed from
\(A_{r,a}(p;\mu_0)-\varepsilon\|P_\Omega D_{r,a}(p)\|_2\ge0\).
Increasing \(\varepsilon\) produces a nested family of smaller feasible sets.
For sufficiently large $\varepsilon$, the nominal optimizer becomes infeasible.

Let \(p^\star(0)\) denote the nominal optimizer, and choose
\(
\varepsilon_2 \approx\,\rho\bigl(p^\star(0)\bigr).
\)
Then \(p^\star(0)\) remains feasible up to approximately
\(\varepsilon_2\). For \(\varepsilon_3>\varepsilon_2\), the nominal
optimizer is no longer feasible, while the robust obedient set
remains nonempty. The corresponding optimizer \(p^\star(\varepsilon_3)\)
is marked in the figure.

\begin{figure}[t]
\centering
\includegraphics[width=0.75\columnwidth]{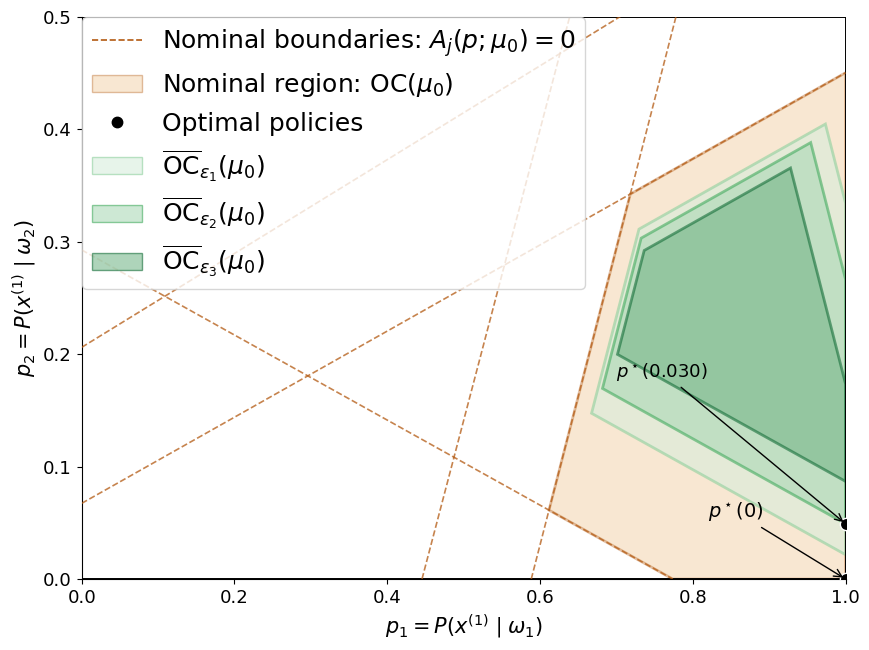}
\caption{
Obedience boundaries (orange dashed), the nominal feasible region
$\OC(\mu_0)$, and shrinking robust regions
$\overline{\OC}_{\varepsilon}(\mu_0)$ (green) for increasing
$\varepsilon$. The nominal optimizer $p^\star(0)$ becomes
infeasible beyond $\varepsilon_2\approx \rho(p^\star(0))$,
and the optimizer at a larger robustness level is also shown.
}
\label{fig:robust_geometry}
\vspace{-0.5em}
\end{figure}
\subsection{Monte Carlo example}
\label{subsec:simulation_cost}

We next investigate sensitivity of the robust information design problem numerically by studying how
the robust value \(V^\star(\varepsilon)\) changes with \(\varepsilon\)
across randomly generated instances with \(|E|=5\), \(|\Omega|=5\), and
\(\mu_0(\omega)=\tfrac15\) for all \(\omega\in\Omega\). The
recommendation profile set $X$ consists of one uniform profile and five
edge-concentrated profiles:
\begin{align*}
x^{(1)}=(0.2,0.2,0.2,0.2,0.2),&
\hspace{1.5pt} x^{(2)}=(0.6,0.1,0.1,0.1,0.1),\\
x^{(3)}=(0.1,0.6,0.1,0.1,0.1),&
\hspace{1.5pt} x^{(4)}=(0.1,0.1,0.6,0.1,0.1),\\
x^{(5)}=(0.1,0.1,0.1,0.6,0.1),&
\hspace{1.5pt} x^{(6)}=(0.1,0.1,0.1,0.1,0.6).
\end{align*}

Using the affine latency model from \eqref{eq:affine_latency},
we generate structured instances in which state \(\omega_e\) favors edge \(e\):
the favored edge is assigned lower slope and intercept ranges, while
the remaining edges are assigned larger ranges. For each instance and each
$\varepsilon$, we solve the robust design problem over the state-profile probabilities
\(\pi(x^{(k)}\mid\omega)\), subject to the simplex constraints and the
robust obedience constraints for all $(r,a)\in\mathcal J$, using
\eqref{eq:realized_cost_difference}.
Over the radius range above, we impose
\(
A_{r,a}(\pi;\mu_0)-\varepsilon
\|P_\Omega D_{r,a}(\pi)\|_2\ge0,
\)
which is second-order-cone representable.


We generate \(100\) instances and retain those that are feasible at
\(\varepsilon=0\). For each retained instance \(i\), let
\(V_i^\star(\varepsilon)\) denote the optimal value at radius $\varepsilon$.
Fig.~\ref{fig:excess_robust_cost} plots
the difference between the value of the nominal information design problem and the robust information design problem
across feasible instances. In blue are the individual value functions trajectories (note their termination when $\overline{\OC}_\varepsilon(\mu_0)$ becomes empty), and overlaid are statistical measures of the data at each $\varepsilon$ including boxplots, mean,
and median are shown. Consistent with the sensitivity discussion above, the excess robust cost is increasing in \(\varepsilon\), and the dispersion across instances also grows.
This effect is reduced in aggregate as problems become infeasible.

\begin{figure}[t]
\centering
\includegraphics[width=0.80\columnwidth]{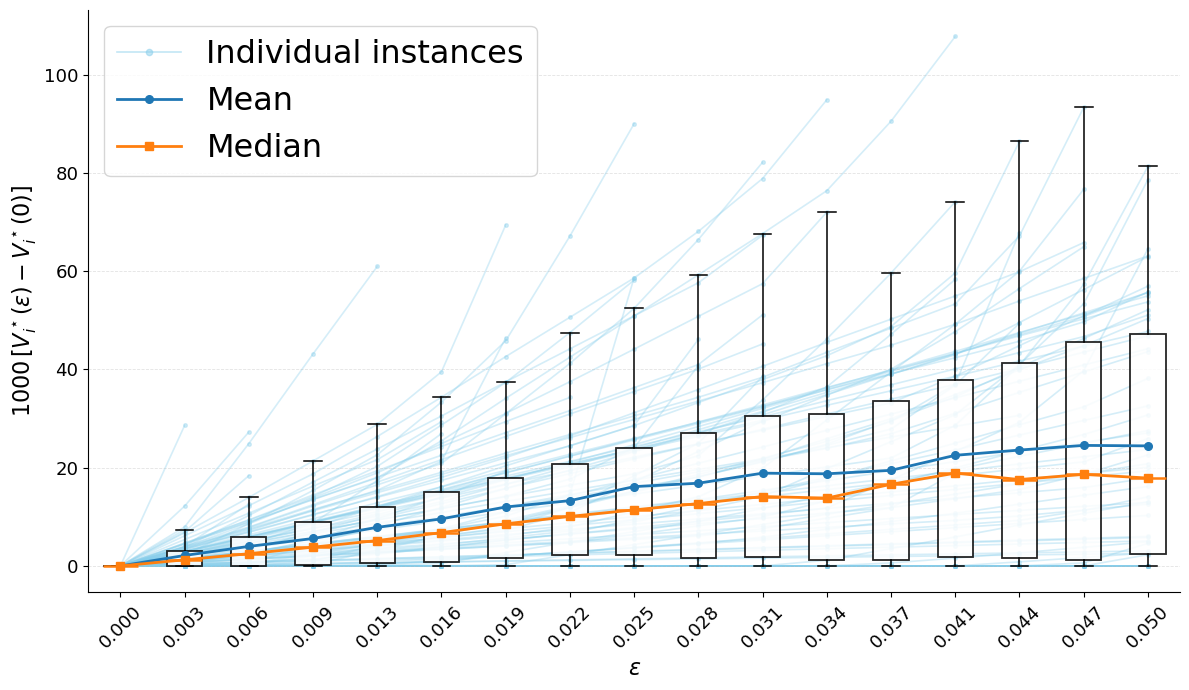}
\caption{
Excess robust cost $\Delta V_i^\star(\varepsilon)$ (scaled by $1000$)
across feasible instances. Light blue curves show individual
trajectories; boxplots, mean, and median are also shown.
}
\label{fig:excess_robust_cost}
\vspace{-0.5em}
\end{figure}
\section{Conclusion}

We studied robust obedience in information design for Bayesian congestion
games under heterogeneous beliefs. We characterized policy-level
robustness ranges, identified regimes where the robust obedience region
remains nonempty, and analyzed the sensitivity of the value function.
When \(\overline{\OC}_\varepsilon(\mu_0)=\emptyset\), no signaling
policy is obedient for every belief in \(U_\varepsilon(\mu_0)\); this does
not rule out obedience for individual beliefs within that neighborhood.
Future work will consider more general network topologies, where the
recommendation space is no longer described by the simple simplex
structure available in parallel networks, as well as sharper
characterizations of exact robustness limits beyond the certified bounds
developed here.

\bibliography{references} 
\bibliographystyle{IEEEtran}

\appendix\label{appendix:finite_support_proof}
\noindent\emph{Proof of Proposition~\ref{prop:heterogeneous_closure_short}:}
Since \(\pi\in \overline{\OC}_\varepsilon(\mu_0;X)\), we have
\(A_{r,a}(\pi;\mu,X)\ge 0\) for all \((r,a)\in\mathcal J\) and
all \(\mu\in U_\varepsilon(\mu_0)\).
Fix any agent $\lambda$. Because $\mu^\lambda\in U_\varepsilon(\mu_0)$,
we have
\(
A_{r,a}(\pi;\mu^\lambda,X)\ge 0,
\qquad
\forall (r,a)\in\mathcal J.
\)
Thus agent $\lambda$ is obedient under $\pi$. Since this holds for
every agent, the aggregate flow generated by the heterogeneous
population is obedient as well.
\hfill$\qed$

\noindent\emph{Proof of Proposition~\ref{prop:finite_support_reduction}:}
For each state \(\omega\in\Omega\) and recommendation profile \(x\in X\),
define \(c^\omega(x):=\sum_{e\in E} x_e\,\ell_e^\omega(x_e)\) and, for
each \((r,a)\in\mathcal J\),
\(
g_{r,a}^\omega(x):=
x_r\bigl(\ell_a^\omega(x_a)-\ell_r^\omega(x_r)\bigr).
\)
Then
\begin{align*}
C(\pi;X)
&=
\sum_{\omega\in\Omega}\mu_0(\omega)\int_X c^\omega(x)\,d\pi_\omega(x),\\
A_{r,a}(\pi;\mu)
&=
\sum_{\omega\in\Omega}\mu(\omega)\int_X g_{r,a}^\omega(x)\,d\pi_\omega(x),
\quad
(r,a)\in\mathcal J.
\end{align*}

Thus, for each fixed \(\omega\), the contribution of \(\pi_\omega\) to
the objective and all obedience constraints is fully determined by the
expected value of
\[
\Gamma^\omega(x)
:=
\Bigl(c^\omega(x),\,(g_{r,a}^\omega(x))_{(r,a)\in\mathcal J}\Bigr)
\in \mathbb{R}^{1+|\mathcal J|}.
\]
Since the latency functions are continuous, \(\Gamma^\omega\) is
continuous. Hence \(\int_X \Gamma^\omega(x)\,d\pi_\omega(x)\) lies in
the convex hull of \(\Gamma^\omega(X)\). By Carath\'eodory's theorem,
there exist points \(x^{\omega,1},\dots,x^{\omega,L_\omega}\in X\),
with \(L_\omega\le |\mathcal J|+2\), and weights
\(\lambda_{\omega,1},\dots,\lambda_{\omega,L_\omega}\ge 0\) satisfying
\(\sum_{j=1}^{L_\omega}\lambda_{\omega,j}=1\), such that
\[
\int_X \Gamma^\omega(x)\,d\pi_\omega(x)
=
\sum_{j=1}^{L_\omega}\lambda_{\omega,j}\Gamma^\omega(x^{\omega,j}).
\]

Let \(\widetilde\pi_\omega\) be the finitely supported probability
measure that places mass \(\lambda_{\omega,j}\) on \(x^{\omega,j}\),
\(j=1,\dots,L_\omega\). Then
\begin{align*}
\int_X c^\omega(x)\,d\widetilde\pi_\omega(x)
&=
\int_X c^\omega(x)\,d\pi_\omega(x),\\
\int_X g_{r,a}^\omega(x)\,d\widetilde\pi_\omega(x)
&=
\int_X g_{r,a}^\omega(x)\,d\pi_\omega(x),
\quad (r,a)\in\mathcal J.
\end{align*}
Applying this construction for each \(\omega\in\Omega\) and taking the
union of the resulting supports yields a finite set
\(\widetilde X\subseteq X\) with
\(\widetilde\pi_\omega\in\Delta(\widetilde X)\) for all \(\omega\in\Omega\).
The preserved statewise cost and deviation terms then imply
\(C(\widetilde\pi;\widetilde X)=C(\pi;X)\) and
\(A_{r,a}(\widetilde\pi;\mu)=A_{r,a}(\pi;\mu)\) for all
\((r,a)\in\mathcal J\) and \(\mu\in\Delta(\Omega)\). This is a standard
Carath\'eodory-type finite-support argument; see, e.g.,
\cite{boyd2004convex,kamenica2011bayesian}.
\hfill\qed


\end{document}